\documentclass[11pt, letterpaper]{article}
\usepackage[utf8]{inputenc}
\usepackage{amsmath,amsfonts,amssymb,amsthm}

\usepackage[colorlinks=true,linkcolor=red,citecolor=blue,urlcolor=red]{hyperref}

\usepackage{mathptmx}
\usepackage[text={6.5in,9in}]{geometry}

\newenvironment{proofof}[1]{\smallskip
\noindent {\bf Proof of #1.  }}{\hfill$\Box$
\smallskip}

\newenvironment{reminder}[1]{\smallskip
\noindent {\bf Reminder of #1  }\em}{}

\def \sgn {\textrm{\it sign}}

\def\poly{\text{poly}}

\def \AM {\text{\sf AM}}

\def \eps {\varepsilon}
\def \R {{\mathbb R}}
\def \Z {{\mathbb Z}}
\def \F {{\mathbb F}}
\def \Q {{\mathbb Q}}
\def \P {{\sf P}}
\def \PH {{\sf PH}}
\def \BP {{\sf BP}}

\def \AC {{\sf AC}}
\def\THR {{\sf LTF}}
\def \NP {{\sf NP}}
\def \BPP {{\sf BPP}}

\def \MOD {{\sf MOD}}
\def \AND {{\sf AND}}

\def \XOR {{\sf XOR}}
\def \MAJ {{\sf MAJ}}
\def \SYM {{\sf SYM}}

\newcommand{\ip}[2]{\ensuremath{\left<#1,#2\right>}}
\def \E {{\sf E}}

\title{Probabilistic Rank and Matrix Rigidity}

\author{Josh Alman\footnote{Computer Science Department, Stanford University, \texttt{jalman@cs.stanford.edu}. Supported by NSF DGE-114747.} \and 
Ryan Williams\footnote{Computer Science Department, Stanford University. Supported in part by a Microsoft Research Faculty Fellowship and NSF CCF-1552651 (CAREER). Any opinions, findings and conclusions or recommendations expressed in this material are those of the authors and do not necessarily reflect the views of the National Science Foundation.} 
}

\date{}

\begin{document}

\newtheorem{fact}{Fact}[section]
\newtheorem{rules}{Rule}[section]
\newtheorem{conjecture}{Conjecture}[section]
\newtheorem{theorem}{Theorem}[section]
\newtheorem{hypothesis}{Hypothesis}
\newtheorem{remark}{Remark}
\newtheorem{proposition}{Proposition}
\newtheorem{corollary}{Corollary}[section]
\newtheorem{lemma}{Lemma}[section]
\newtheorem{claim}{Claim}
\newtheorem{definition}{Definition}[section]

\maketitle

\begin{abstract} 
We consider a notion of probabilistic rank and probabilistic sign-rank of a matrix, which measures the extent to which a matrix can be probabilistically represented by low-rank matrices. We demonstrate several connections with matrix rigidity, communication complexity, and circuit lower bounds. The most interesting outcomes are:

\smallskip

{\bf The Walsh-Hadamard Transform is \emph{Not} Very Rigid.} We give surprising \emph{upper bounds} on the rigidity of a family of matrices whose rigidity has been extensively studied, and was conjectured to be highly rigid. For the $2^n \times 2^n$ Walsh-Hadamard transform $H_{n}$ (a.k.a. Sylvester matrices, a.k.a. the communication matrix of Inner Product modulo $2$), we show how to modify only $2^{\eps n}$ entries in each row and make the rank of $H_{n}$ drop below $2^{n(1-\Omega(\eps^2/\log(1/\eps)))}$, for all small $\eps > 0$, over any field. That is, it is \emph{not possible} to prove arithmetic circuit lower bounds on Hadamard matrices such as $H_{n}$, via L.~Valiant's matrix rigidity approach. We also show non-trivial rigidity upper bounds for $H_n$ with smaller target rank. 

\smallskip

{\bf Matrix Rigidity and Threshold Circuit Lower Bounds.} We give new consequences of rigid matrices for Boolean circuit complexity. First, we show that explicit $n \times n$ Boolean matrices which maintain rank at least $2^{(\log n)^{1-\delta}}$ after $\frac{n^2}{2^{(\log n)^{\delta/2}}}$ modified entries (over \emph{any} field, for any $\delta > 0$) would yield an explicit function that does not have sub-quadratic-size $\AC^0$ circuits with two layers of arbitrary linear threshold gates.
Second, we prove that explicit 0/1 matrices over $\R$ which are modestly more rigid than the best known rigidity lower bounds for \emph{sign-rank} would imply \emph{exponential-gate} lower bounds for the infamously difficult class of depth-two linear threshold circuits with arbitrary weights on both layers ($\THR \circ \THR$). In particular, we show that matrices defined by these seemingly-difficult circuit classes actually have low probabilistic rank and sign-rank, respectively.

\smallskip

{\bf An Equivalence Between Communication, Probabilistic Rank, and Rigidity.} It has been known since Razborov [1989] that explicit rigidity lower bounds would resolve longstanding lower-bound problems in communication complexity, but it seemed possible that communication lower bounds could be proved without making progress on matrix rigidity. We show that for every function $f$ which is randomly self-reducible in a natural way (the inner product mod $2$ is an example), bounding the communication complexity of $f$ (in a precise technical sense) is \emph{equivalent} to bounding the rigidity of the matrix of $f$, via an equivalence with probabilistic rank.
\end{abstract}

\thispagestyle{empty}
\newpage
\setcounter{page}{1}

\section{Introduction}

Let $R$ be a ring. In analogy with the notion of a probabilistic polynomial, we define a \emph{probabilistic matrix over $R$} to be a distribution of matrices ${\cal M} \subset R^{n \times n}$. A probabilistic matrix ${\cal M}$ \emph{computes} a matrix $A \in R^{n \times n}$ with error $\eps > 0$ if for every entry $(i,j) \in [n]^2$,
\[\Pr_{B \sim {\cal M}}[A[i,j] = B[i,j]] \geq 1-\eps.\] In this way, a probabilistic matrix is a \emph{worst-case randomized representation} of a fixed matrix. A probabilistic matrix ${\cal M}$ has rank $r$ if the maximum rank of a $M \sim {\cal M}$ is $r$.

We define the \emph{$\eps$-probabilistic rank} of a matrix $M \in R^{n \times n}$ to be the minimum rank of a probabilistic matrix computing $M$ with error $\eps$. Such probabilistic matrices are of interest and potentially very useful, because some full rank matrices can be represented by probabilistic matrices of rather low rank. For example, every identity matrix has $\eps$-probabilistic rank $O(1/\eps)$ over any field, by simulating a protocol for EQUALITY using $\log(1/\eps) + O(1)$ communication that computes random inner products (cf. Theorem~\ref{lem:eq-rank}).

We began studying probabilistic rank in the hopes of better understanding the use of probabilistic polynomials in algorithm design. Recent work has shown how substituting low-degree probabilistic polynomials in place of common subroutines can be very useful for speeding up the best known running times for many core problems~\cite{Williams14a,Williams14f,AbboudWY15,JoshRyan,Alman-Chan-Williams16,SystemsofPolyEqns}. However, almost every algorithmic application ends up embedding the low-degree polynomial evaluation problem in a \emph{fast multiplication of two low-rank (rectangular) matrices}. That is, this algorithmic work is really using the fact that that various circuits and subroutines from core algorithms have \emph{low probabilistic rank}, and is applying low-rank representations to obtain an algorithmic speedup. Because ``low probabilistic rank'' is potentially a far broader notion than that of ``low-degree probabilistic polynomials'', it makes more sense to study probabilistic rank directly, in the hopes of finding stronger algorithmic applications.

In this paper, we consider complexity-theoretic aspects of probabilistic rank. We demonstrate how probabilistic rank is a powerful notion for understanding the age-old problem of matrix rigidity, and some models of communication complexity where knowledge is still sparse.

\paragraph{Matrix Rigidity.} A central part of our paper connects the probabilistic rank of a matrix to its rigidity. 
The \emph{rank-$r$ rigidity} of a matrix $A \in R^{n\times n}$, denoted by ${\cal R}_A(r)$, is the minimum Hamming distance from $A$ to an $n \times n$ matrix of rank $r$ over $R$. 
That is, ${\cal R}_A(r)$ is the number of entries of $A$ that must be modified in order for the rank to drop to $r$. 
(Sometimes we'll want to work over a particular field $K$; in that case we'll speak of ``${\cal R}_A(r)$ over $K$.'')  Matrix rigidity was introduced by Leslie Valiant~\cite{Valiant77} in 1977, as a path towards arithmetic circuit lower bounds for linear transformations. Valiant showed that for a field $\F$, and every linear transformation $T : \F^n \rightarrow \F^n$ computable by a circuit of $O(n)$ addition gates of bounded fan-in (with scalar multiplications on the wires) and $O(\log n)$ depth, ${\cal R}_T(O(n/\log \log n)) \leq n^{1+\eps}$, for every fixed $\eps > 0$. Thus to prove a circuit lower bound for $T$, it suffices to lower bound the rigidity of $T$ for rank $O(n/\log \log n)$. Valiant proved that random 0/1 matrices over a field are highly rigid (whp), and strong rigidity lower bounds are known when one allows exponential (or infinite) precision in the matrix entries~\cite{Lokam06,Kumar2014}. 
However, no explicit rigid matrices $T$ with (say) ${\cal R}_T(O(n/\log \log n)) > n^{1.0001}$ are known\footnote{An infinite family of matrices $\{M_n \mid n \in S\}$ is said to be explicit if there is a polynomial time algorithm $A$ such than $A(1^n)$ prints $M_n$ when $n \in S$.}, despite decades of effort (see the surveys~\cite{Codenotti2000a,Lokam09} and the recent work~\cite{GoldreichT16}). The best known lower bounds for explicit $M$ yield only ${\cal R}_M(r) \geq \Omega(\frac{n^2}{r} \cdot \log(n/r))$~\cite{Friedman93,Shokrollahi1997}, and Lokam~\cite{Lokam2000} argues that known methods (``untouched minor arguments'') cannot prove rigidity lower bounds larger than this.
Very recently, Goldreich and Tal showed an improved ``semi-explicit'' rigidity lower bound: for random $n \times n$ Toeplitz matrices $M$, they proved that ${\cal R}_M(r) \geq \Omega(\frac{n^3}{r^2 \log n})$ whp, when rank $r \geq \sqrt{n}$~\cite{GoldreichT16} (note that such matrices can be generated with $\tilde{O}(n)$ bits of randomness).

In 1989, Razborov~\cite{Razborov89} (see also \cite{Wunderlich12}) described a connection between matrix rigidity and communication complexity: Letting $f$ be a function in $\PH^{cc}$ (the communication complexity equivalent of the polynomial-time hierarchy), the $2^n \times 2^n$ communication matrix $M_f$ of $f$ has ${\cal R}_{M_f}(2^{\log^c(n/\eps)}) \leq \eps \cdot 4^n$, where $\eps > 0$ is arbitrary and $c > 0$ is a constant depending only on $f$, but not $n$. (Razborov's proof uses low-degree polynomials which approximate $\AC^0$ functions.) Thus, explicit rigidity lower bounds in the ``low'' rank and ``high'' error setting would imply long-open communication lower bounds. 

Among the many attempts to prove arithmetic circuit lower bounds via rigidity, perhaps the most commonly studied explicit matrix has been the Walsh-Hadamard transform~\cite{PudlakS88,Alon90,Grigor,Nisan,KashinR98,Codenotti2000a,Lokam2001,LandsbergTV03,Mid05,Wolf06,Rashtchian16}:

\begin{definition} For vectors $x,y \in \R^d$, let $\langle x,y \rangle$ denote their inner product. Let $v_1,\ldots,v_{2^n} \in \{0,1\}^n$ be the enumeration of all $n$-bit vectors in lexicographical order. The \emph{Walsh-Hadamard matrix} $H_n$ is the $2^n \times 2^n$ matrix defined by $H_n(v_i,v_j) := (-1)^{\langle v_i,v_j \rangle}$. 
\end{definition}

It was believed that $H_n$ is rigid because its rows are mutually orthogonal (i.e., $H_n$ is Hadamard), so in several of the above references, only that property was assumed of the matrices. The best rigidity lower bounds known for $H_n$ have the form ${\cal R}_{H_n}(r) \geq \Omega(4^n/r)$; for the target rank $r = O(2^n/\log n)$ in Valiant's problem, the lower bound is only $\Omega(2^n \log n)$. It was a folklore theorem that one can modify only $O(n)$ entries of an $n \times n$ Hadamard matrix and make its rank at most $n/2$~\cite{Satya-hw}, but it was believed that for lower rank many more entries would require modification.

\paragraph{Hadamard Ain't So Rigid.} We give a good excuse for the weakness of these lower bounds:

\begin{theorem}[Non-Rigidity of Hadamard Matrices]\label{nonrigid-IP2} For every field $K$, for every sufficiently small $\eps > 0$, and for all $n$, we have ${\cal R}_{H_n}\left(2^{n-f(\eps) n}\right) \leq 2^{n(1+\eps)}$ over $K$, for a function $f$ where $f(\eps)=\Theta(\eps^2/\log(1/\eps))$.  
\end{theorem}

In fact, we show a strong non-rigidity upper bound: by modifying at most $2^{\eps n}$ entries in each row of $H_n$, the rank of $H_n$ drops to $2^{n-f(\eps) n}$. That is, the matrix rigidity approach to arithmetic circuit lower bounds \emph{does not} apply to Hadamard matrices such as the Walsh-Hadamard transform. We would have required lower bounds of the form ${\cal R}_{H_n}(2^{n}/(\log n)) \geq 2^{n(1+\eps)}$ for some $\eps > 0$ to obtain circuit lower bounds; the upper bound of Theorem~\ref{nonrigid-IP2} shows this is impossible. The proof is in Section~\ref{non-rigidity-WH}.

We do not (yet) believe that the Walsh-Hadamard transform has $O(2^n)$-size $O(n)$-depth circuits; a more appropriate conclusion is that rigidity is too coarse to adequately capture the lower bound problem in this case. Having said that, Theorem~\ref{nonrigid-IP2} does imply new circuit constructions: it follows that there is a depth-two unbounded fan-in arithmetic circuit for the Walsh-Hadamard transform with $2^{n+O(\eps\log(1/\eps)) n} + 2^{2n-\Omega(\eps^2 n)}$ gates; setting $\eps > 0$ appropriately, we have a $4^{\delta n}$-size circuit for some $\delta < 1$.

We also show non-trivial rigidity upper bounds for $H_n$ in the regime that would be useful for communication complexity, where the rigidity is much closer to $4^n$. 

\begin{theorem}[Non-Rigidity of Hadamard Matrices, Part II]\label{high-error-IP2} For every integer $r \in [2^{2n}]$, one can modify at most $2^{2n}/r$ entries of $H_n$ and obtain a matrix of rank $(n/\ln(r))^{O(\sqrt{n \log(r)})}$.
\end{theorem}

See Appendix~\ref{sec:high-error-WH} for the proof. While the product of rank and rigidity (a natural measure) of $H_n$ is only known to be at least $\Omega(4^n)$, Theorem~\ref{high-error-IP2} provides an upper bound of $4^n \cdot n^{O\left(\sqrt{n \log(r)}\right)}/r$, which is  not small enough to refute the conjectured rigidity lower bounds required for communication complexity applications. But as we show later, these upper bounds still have non-trivial consequences for the communication complexity of IP2. 

\paragraph{New Applications of Explicit Rigid Matrices.} Rigidity has been studied primarily for its connections to communication complexity and to lower bounds on arithmetic circuits computing linear transformations. We show new implications of constructing explicit rigid matrices for Boolean circuit complexity.

First, we show how explicit rigidity lower bounds would yield Boolean circuit lower bounds where only somewhat weak results are known:

\begin{theorem}\label{THR-circuits-rigidity} Let $K$ be an arbitrary field, and $\{M_n\}$ be a family of Boolean matrices such that (a) $M_n$ is $n\times n$, (b) there is a $\poly(\log n)$ time algorithm $A$ such that $A(n,i,j)$ prints $M_n(i,j)$, and (c) there is a $\delta > 0$ such that for infinitely many $n$, 
\vspace{-5mm}
\[{\cal R}_{M_n}\left(2^{(\log n)^{1-\delta}}\right) \geq \frac{n^2}{2^{(\log n)^{\delta/2}}} \text{ over $K$}.\] 
Then the language $\{(n,i,j)\mid M_n(i,j)=1\} \in \P$ does not have $\AC^0 \circ \THR \circ \AC^0 \circ \THR$ circuits of $n^{2-\eps}$-size and $o(\log n/\log \log n)$-depth, for all $\eps > 0$.
\end{theorem}

The theorem is obtained by giving non-trivial probabilistic rank bounds for such circuits, building on Lokam~\cite{Lokam2001}. Therefore, proving rigidity (or probabilistic rank) lower bounds for explicit 0/1 matrices over a field $K$ would imply nearly-quadratic size lower bounds for $\AC^0 \circ \THR \circ \AC^0 \circ \THR$ circuits of unbounded depth, a powerful class of Boolean circuits. (The best known lower bounds are that functions in the huge class $\E^{\NP}$ do not have such circuits~\cite{Alman-Chan-Williams16}.) See Appendix~\ref{sec:rigid-circuit-lbs} for these results.

\paragraph{Sign-Rank Rigidity.} The sign rank of a $-1/1$ matrix $M$ is the lowest rank of a matrix $N$ such that $\text{sign}(M[i,j])=\text{sign}(N[i,j])$, for all $(i,j)$. Lower bounds on the sign-rank of matrices were used 15 years ago to prove exponential lower bounds against $\THR \circ \MAJ$ and $\THR \circ \SYM$ circuits~\cite{forster02,ForsterKLMSS01}, i.e. restricted versions of depth-two threshold circuits. We extend the sign-rank connection to a circuit class for which strong lower bounds have long been open: explicit matrices with high rigidity under sign-rank would imply strong depth-two threshold circuit lower bounds. (Here, sign-rank rigidity is defined in the natural way, with ``rank'' replaced with ``sign-rank'' in the rigidity definition.)

A corollary of a theorem of Razborov and Sherstov~\cite{Razborov-Sherstov10} (see Appendix~\ref{sec:sign-rank}) is that for all $n$, $H_n$ has sign-rank $r$-rigidity at least $\Omega(4^n/r)$, just as in the case of normal rank rigidity. We show that even a somewhat minor improvement would already imply exponential-size lower bounds for depth-two linear threshold circuits with unbounded weights on both layers, a problem open for decades~\cite{Hajnal93,KaneW16}:

\begin{theorem} Suppose the sign rank $r$-rigidity of $H_n$ is $\Omega(4^n/r^{.999})$ for some rank bound $r \geq 2^{\alpha n}$ and some $\alpha > 0$. Then the Inner Product Modulo $2$ requires $2^{\Omega(n)}$-size $\THR \circ \THR$ circuits.
\end{theorem}

Theorem~\ref{sign-rank-rigid-IP2} gives a more general statement. Under the hood is an upper bound: matrices defined by small $\THR \circ \THR$ circuits have low \emph{probabilistic sign-rank}: for every such circuit of $s$ gates, viewing its truth table as a $2^{n/2} \times 2^{n/2}$ matrix, there is a distribution of $O(s^2 n^2/\eps)$-rank matrices which sign-represent the truth table in a worst-case probabilistic way with error $\eps$.

\paragraph{Rigidity, Communication, and Probabilistic Rank: An Equivalence.} Probabilistic rank arises very naturally in studying generalized models of communication complexity. For a Boolean function $f : \{0,1\}^n \times \{0,1\}^n \rightarrow \{0,1\}$, let $M_f$ be the $2^n \times 2^n$ \emph{truth table matrix of $f$} with $M_f[x,y] = f(x,y)$ for all $x,y$. The following correspondence between probabilistic rank and communication complexity is immediate (one could even take the proposition as a \emph{definition} of $\BP \cdot \MOD_m \P$ communication complexity). 

\begin{proposition}\label{rank-equals-cc} Let $m > 1$ be an integer, let $f : \{0,1\}^n \times \{0,1\}^n \rightarrow \{0,1\}$, and let $M_f$ be its truth table matrix. The $\BP \cdot \MOD_m \P$ communication complexity of $f$ with error $\eps$ equals the (base-2) logarithm of the $\eps$-probabilistic rank of $M_f$ over $\Z_m$ (within additive constants). 
\end{proposition}

Similarly, $\AM$ (Arthur-Merlin communication complexity) is equivalent to probabilistic Boolean rank. 

It's easy to see that if a matrix has $\eps$-probabilistic rank $r$, then its rank-$r$ rigidity is at most $\eps 2^{2n}$; thus rigidity lower bounds imply communication lower bounds. But conversely, it seems easier to prove lower bounds on probabilistic rank compared to rigidity: with probabilistic rank, we need to rule out a ``distribution'' of erroneous matrix entries which are required to ``spread the errors'' around; with rigidity, we have to rule out \emph{any} adversarial choice of bad entries. 

We show (in Appendix~\ref{sec:general-self-reduction}) that for every randomly self-reducible function $f : \{0,1\}^{2n} \rightarrow R$ in which the self-reduction makes $k$ non-adaptive queries, low rigidity implies low probabilistic rank: the $\eps$-probabilistic rank of its corresponding matrix is at most $(kr)^k$ if its rank-$r$ rigidity is at most $\eps \cdot 4^n$. Thus there is a strong relationship between $\eps$-probabilistic rank (and communication complexity, by Proposition~\ref{rank-equals-cc}) and the rank for which the rigidity is an $\eps$-fraction of the matrix. For the Walsh-Hadamard transform, we prove (in Section~\ref{sec:equivalence}) that the probabilistic rank of $H_n$ and the rigidity of $H_n$ are equivalent concepts over fields:

\begin{theorem}\label{prob-rank-rigidity-IP2} For every field $K$ and for every $n$, ${\cal R}_{H_n}(r) \leq \eps \cdot 4^n$ over $K$ if and only if $H_n$ has $\eps$-probabilistic rank $r$ over $K$.
\end{theorem}

The matrices $H_n$ represent the communication matrices of the widely-studied Inner Product Modulo $2$ (IP2) function. By Proposition~\ref{rank-equals-cc}, the $\BP \cdot \MOD_p \P$ communication complexity of IP2 and the rigidity of $H_n$ over $\F_p$ are really equivalent concepts.  Applying this theorem, our earlier rigidity upper bounds also imply some modest but interesting improvements on communication complexity protocols.  From the rigidity upper bound of Theorem~\ref{high-error-IP2}, we obtain a communication protocol for IP2 with 
$O(\sqrt{n \log(1/\eps)}\log(\frac{n}{\log(1/\eps)}))$ bits and error $\eps$ in the $\BP \cdot \MOD_p \P$ communication model, for every prime $p$. (Aaronson and Wigderson gave an ${\sf MA}$ protocol for IP with $O(\sqrt{n} \log(n/\eps))$ communication complexity and error $\eps$~\cite{AW09}; ours is more efficient for $\eps \ll 1/2^{\sqrt{\log n}}$.) Applying Theorem~\ref{nonrigid-IP2} yields an IP2 protocol with $n(1-\Omega(\eps^2/\log(1/\eps)))$ communication and only $1/2^{n-\eps n}$ error. We are skeptical that our rigidity upper bounds for $H_n$ are tight; we hope these results will aid future work (to prove rigidity upper bounds, one only has to think about communication protocols for IP2).

\subsection{Related Work}

Besides the many references already mentioned earlier,  there are a few other related works we know of.

\noindent{\bf Toggle Rank.} By Yao's minimax principle~\cite{Yao83}, $\BP \cdot \MOD_m \P$ communication complexity (randomized communication with ``counting modulo $m$'' power) equals worst-case distributional $\MOD_m \P$ communication complexity. In matrix terms, putting an arbitrary distribution ${\cal P}$ on the pairs $\{0,1\}^n \times \{0,1\}^n$, the worst-case $\eps$-distributional complexity of $M$ is the lowest rank (over $\Z_m$) of a $2^n \times 2^n$ matrix $N$ with error $||M-N||\leq \eps$ over ${\cal P}$. Wunderlich \cite{Wunderlich12} calls this rank notion the \emph{approximate toggle rank}. 
Proposition~\ref{rank-equals-cc} shows that probabilistic rank and approximate toggle rank are very closely related, but they are \emph{\bf \emph{not}} the same as the usual rigidity concept, which corresponds to the uniform distribution on pairs. For structured functions like IP2, we prove (Theorem~\ref{prob-rank-rigidity-IP2}) that the uniform distribution \emph{is} the worst case.

\noindent{\bf Sign-rank Rigidity and AC0-MOD2 circuits.} A tantalizing open problem that has gained popularity in recent years~\cite{ServedioV12,AkaviaBGKR14,CheraghchiGJWX16} is whether IP2 has polynomial-size $\AC^0 \circ \MOD_2$ circuits: i.e., circuits of $O(1)$-depth over AND/OR/NOT, but with a layer of gates computing PARITY at the bottom nearest the inputs. Servedio and Viola~\cite{ServedioV12} propose an interesting attack: in our terminology, they note that $\AC^0 \circ \MOD_2$ circuits of size $s$ have $n^{O(\log^{d-1} s)\log(1/\eps)}$ sign-rank rigidity at most $\eps 2^{2n}$ over $\R$, and prove a lower bound on the correlation of signs of sparse polynomials (taken as a proxy for low-rank sign-matrices) with IP2. That is, they prove a weak sign-rank rigidity lower bound (note Razborov and Sherstov prove an analogous lower bound for sign-rank rigidity of IP2; see Appendix~\ref{sec:sign-rank}). Our results have two consequences for this sort of approach. First, Theorem~\ref{thm:thrthr-sign-rank} shows that a sign-rank rigidity lower bound would prove something much stronger: a lower bound for \emph{depth-two threshold circuits} computing IP2, a longstanding open problem. Second, our non-trivial upper bounds on the rank rigidity of the IP2 matrix (which is $H_n$) suggest that IP2 may have much lower sign-rank rigidity than expected. 

\noindent{\bf Sign-Rank Rigidity and Margin Complexity.} 
Linial and Shraibman~\cite{LinialS09} 
prove (in our terminology) that the sign-rank rigidity of an $n \times n$ matrix $A$ is at most $\eps n^2$ for target rank $O(mc(A)^2 \log(1/\eps))$, where $mc(A)$ is the ``margin complexity'' of $A$. Thus the margin complexity of a matrix can be used to upper bound sign-rank rigidity. They also study rigidity notion based on $mc$, conjecture that high $mc$ implies high margin-complexity rigidity, and show that high margin-complexity also implies communication complexity lower bounds (for similar parameters as the standard rank-rigidity setting).

\noindent{\bf Approximate Rank.} A different ``approximating'' rank notion has been studied in~\cite{buhrman2001communication,klivans2010lower,alon2013approximate}, with connections to quantum computing and approximation algorithms. The $\eps$-approximate rank of $M \in \R^{n\times n}$ is the lowest rank of a matrix $A$ such that $||M-A||_{\infty} \leq \eps$. That is, we can obtain one matrix from the other by perturbing each entry by at most $\eps > 0$. The appropriate analogy here seems to be that probabilistic polynomials are to probabilistic rank, as $\ell_{\infty}$-approximate polynomials are to approximate rank: both are natural generalizations of polynomial representations to matrix representations, with different properties. 

\section{Preliminaries}

We assume basic familiarity with complexity theory. For circuit complexity, we use ${\cal C} \circ {\cal D}$ to denote depth-two circuits where the output gate is of type ${\cal C}$ and the ``hidden'' layer is of type ${\cal D}$, e.g., $\THR \circ \THR$ denotes ``depth-two linear threshold circuits'',  $\THR \circ \MOD_2$ denotes ``linear threshold function of parities'', etc. For variables $x_1,\ldots,x_n$, we use $\vec{x}$ to denote $(x_1,\ldots,x_n)$, and for $\vec{x} \in \{0,1\}^n$, we use $|\vec{x}| = \sum_i x_i$ to denote its Hamming weight. We use the Iverson bracket $[P] : \{0,1\}^n \rightarrow \{0,1\}$ to denote the Boolean function which outputs $1$ if and only if property $P$ is true of the $n$ inputs. Below we describe some basic properties relating probabilistic polynomials, probabilistic rank, and rigidity.

\begin{definition}
Let $R$ be any ring, and $f : \{0,1\}^{2n} \to R$ be any function on $2n$ Boolean variables. The \emph{truth table matrix} $M_f$ of $f$ is the $2^n \times 2^n$ matrix given by
$$M_f(v_i, v_j) = f(v_i, v_j),$$
where $v_1,\ldots,v_{2^n} \in \{0,1\}^n$ is the enumeration of all $n$-bit vectors in lexicographical order.
\end{definition}

Given the above definition, it is natural to define the probabilistic rank of a function:

\begin{definition}
The $\eps$-probabilistic rank of a function $f : \{0,1\}^{2n} \to R$ is the $\eps$-probabilistic rank of its truth table matrix $M_f$. The rank of $f$ and the rigidity of $f$ are defined similarly.
\end{definition}

\begin{definition}[Razborov~\cite{Razborov}, Smolensky~\cite{Smolensky87}]
Let $R$ be a ring, and let $f : \{0,1\}^n \to R$. A \emph{probabilistic polynomial} for $f$ with error $\eps$ and degree $d$ is a distribution ${\cal P}$ on polynomials $p : \{0,1\}^n \to R$ of degree at most $d$ such that for \emph{every} $x \in \{0,1\}^n$, $\Pr_{p \sim {\cal P}}[p(x) = f(x)] \geq 1 - \eps$. We may similarly refer to a probabilistic polynomial with $m$ monomials.
\end{definition}

The following simple mapping from sparse polynomials to low-rank matrices is very useful:

\begin{lemma} \label{lemma:monomial-to-rank}
Let $R$ be any ring, and $f : \{0,1\}^{2n} \to R$. Let $p : R^{2n} \to R$ be a polynomial with $m$ monomials such that $p(x,y) = f(x,y)$ for any $x,y \in \{0,1\}^n$. Then the rank of $f$ is at most $m$.
\end{lemma}

\begin{proof}
Let $a_1, \ldots, a_m, b_1, \ldots, b_m : R^n \to R$ be monomials such that $p(x,y) = \sum_{i=1}^m a_i(x) \cdot b_i(y)$ is the monomial expansion of $p$. For $1 \leq i \leq m$, define vectors $\vec{\alpha_i}, \vec{\beta_i} \in R^{2^n}$ by $\vec{\alpha_i}[x] = a_i(x)$ and $\vec{\beta_i}[y] = b_i(y)$ for each $x,y \in \{0,1\}^n$. Then 
$M_f = \sum_{i=1}^m \vec{\alpha_i} \otimes \vec{\beta_i}$, where $\otimes$ denotes the outer product of vectors. Thus $\text{rank}(M_f) \leq m$.
\end{proof}

As a corollary, the probabilistic rank of $f$ is at most the sparsity of a probabilistic polynomial for $f$:

\begin{corollary}
Let $R$ be any ring, and $f : \{0,1\}^{2n} \to R$. If $f$ has a probabilistic polynomial ${\cal P}$ with at most $m$ monomials and error $\eps$, then the $\eps$-probabilistic rank of $f$ is at most $m$.
\end{corollary}

\begin{proof}
Let $p$ be a polynomial in the support of the distribution ${\cal P}$. Since $p$ has at most $m$ monomials, by Lemma \ref{lemma:monomial-to-rank} the truth table matrix $M_p$ of $p$ (restricted to the domain $\{0,1\}^{2n}$) has rank at most $m$. The distribution of $M_p$ over $p$ drawn from ${\cal P}$ is therefore an $\eps$-probabilistic rank-$m$ distribution for $M_f$, since $M_f(x,y) = M_p(x,y)$ if and only if $f(x,y) = p(x,y)$.
\end{proof}

By drawing a `typical' matrix from the probabilistic rank distribution, we can always obtain a matrix rigidity upper bound from a sparse probabilistic polynomial.

\begin{corollary} \label{cor:poly-to-rigid}
Let $R$ be any ring, and $f : \{0,1\}^{2n} \to R$ be any function on $2n$ Boolean variables. If $f$ has a probabilistic polynomial $P$ with at most $m$ monomials and error $\eps$, then one can modify $\eps 2^{2n}$ entries of the truth table matrix $M_f$ and obtain a matrix of rank at most $m$.
\end{corollary}

\section{Non-Rigidity of Walsh-Hadamard}
\label{non-rigidity-WH}

Now we prove that the Walsh-Hadamard matrices are not rigid enough for Valiant's program: 

\begin{reminder}{Theorem~\ref{nonrigid-IP2}} For every field $K$, for every sufficiently small $\eps > 0$, and for all $n$, we have ${\cal R}_{H_n}\left(2^{n-f(\eps) n}\right) \leq 2^{n(1+\eps)}$ over $K$, for a function $f$ where $f(\eps)=\Theta(\eps^2/\log(1/\eps))$.
\end{reminder}

For a vector $v \in \{0,1\}^n$, let $|v|$ be the number of ones in $v$. Let $H : [0,1] \to [0,1]$ denote the binary entropy function $$H(p) = -p \log_2 p - (1-p) \log_2 (1-p).$$
We will need some estimates of binomial coefficients. For $\eps \in (0,1/2)$:
\begin{equation}\label{entropy}
\binom{n}{\eps n} \leq n \cdot 2^{H(\eps) n}, \text{ and}
\end{equation}
\begin{equation}\label{entropy-half}
2^{n-O(\eps^2 n)} \leq \binom{n}{(1/2-\eps) n} \leq 2^{n-\Omega(\eps^2 n)}.
\end{equation}
Equation \eqref{entropy} is standard; equation \eqref{entropy-half} follows from standard tail bounds on the binomial distribution. In particular, the probability that a uniform random bit string has at most $(1/2-\eps)n$ ones is at most $2^{-c_1\eps^2 n}$ and at least $2^{-c_2 \eps^2 n}$, for universal constants $c_1, c_2 > 0$.

Our first (simple) lemma uses a polynomial to compute a large fraction of $H_n$'s entries with a low-rank matrix. However, this fraction won't be high enough; we'll need another idea to ``correct'' many entries later.

\begin{lemma}\label{IP2-polynomial}
For every field $K$, and for every $\eps \in (0,1/2)$, there is a multilinear polynomial $p(x_1, \ldots, x_n, y_1, \ldots, y_n)$ over $K$ with at most $2^{n-\Omega(\eps^2 n)}$ monomials, such that for all $\vec{x},\vec{y} \in \{0,1\}^n$ with $\langle \vec{x},\vec{y}\rangle \in [2 \eps n, (1/2 + \eps) n]$, \[p(\vec{x},\vec{y}) = (-1)^{\langle \vec{x},\vec{y}\rangle}.\]
\end{lemma}

The proof uses properties of multivariate polynomial interpolation over the integers. To be concrete, we will apply the following lemma from one of our previous papers:

\begin{lemma}[\cite{JoshRyan}, Lemma 3.1]\label{IP2-polynomial-interpolation}
For any integers $n,r,k$ with $n \geq r+k$ and any integers $c_1, \ldots, c_r$, there is a multivariate polynomial $p : \{0,1\}^n \to \Z$ of degree $r-1$ with integer coefficients such that $p(z) = c_i$ for all $\vec{z} \in \{0,1\}^n$ with Hamming weight $|\vec{z}| = k+i$.
\end{lemma}

Intuitively, Lemma~\ref{IP2-polynomial-interpolation} is true because the dimension of the space of degree-$(r-1)$ polynomials in $n$ variables is large enough that we can always construct a polynomial with the desired constraints. 

\begin{proof}[Proof of Lemma \ref{IP2-polynomial}] By Lemma \ref{IP2-polynomial-interpolation} with $k=2 \eps n - 1$, $r = (1/2 - \eps)n + 1$, and $c_i = (-1)^{k+i}$, one can construct a multivariate polynomial $q : \{0,1\}^n \to \Z$ with integer coefficients, of degree $(1/2 - \eps)n$, such that for all $\vec{z} \in \{0,1\}^n$ with $|\vec{z}| \in [2 \eps n, (1/2 + \eps) n]$, we have $q(\vec{z}) = (-1)^{|\vec{z}|}$. Since the prime subfield of every field $K$ is either $\Q$ or $\F_{m}$ for some prime $m$, and $q$ has integer coefficients, $q$ can be viewed as a polynomial over $K$ (by taking the coefficients modulo $m$ if appropriate). Then our desired polynomial is \[p(x_1, \ldots, x_n, y_1, \ldots, y_n) = q \left( x_1 y_1, x_2 y_2, \ldots, x_n y_n\right).\]
We can upper-bound the number of monomials in $p$ as follows. First, since we only care about the value of $p$ on $\{0,1\}^{2n}$, we can make $p$ multilinear by applying the equation $v^2 = v$ to all variables. Second, observe that for all $i=1,\ldots,n$, $x_i$ and $y_i$ appear in exactly the same monomials. So if we introduce a variable $z_i$ in place of each $x_i \cdot y_i$ in $p$, the number of monomials in our new $n$-variate polynomial $p'$ equals the number of monomials in $p$. 

Since $p'$ is multilinear and degree $(1/2 - \eps)n + 1$, the number of monomials is at most $n \binom{n}{(1/2 - \eps)n + 1}$, which by \eqref{entropy-half} is at most $2^{n-c_2\eps^2 n}$ for some constant $c_2 > 0$.
\end{proof}

Our second lemma says: fixing a vector $x$ with about $1/2$ ones, there is a strong upper bound the number of vectors which has about $1/2$ ones but has small (integer) inner product with $x$; we'll use this to upper bound the number of erroneous entries at the very end.

\begin{lemma} \label{small-IP} For every vector $x \in \{0,1\}^n$ with $|x| \in [(1/2 -a)n, (1/2+a)n]$, and any parameters $a, b \in (0,1/5)$, the probability that a uniformly random vector $y$ from $\{0,1\}^n$ satisfies both
\begin{itemize}
\item $|y| \in [(1/2 - a)n, (1/2 + a)n]$, and
\item $\sum_{k=1}^n x_k y_k \leq bn$
\end{itemize}
is at most $(2an+1)(bn+1) \cdot 2^{(f(a,b) - 1)n}$, where $f$ is a function such that $f(a,b) \to 0$ as $a,b \to 0$.
\end{lemma}

The usual toolbox of small-deviation estimates does not seem to yield the lemma; we give a direct proof.  

\begin{proof}
For all $x$ of the above form, every $k \in [(1/2 - a)n, (1/2 + a)n]$, and every $s \leq bn$, we count the number of $y \in \{0,1\}^n$ with $|y| = k$ and $\sum_{k=1}^n x_k y_k = s$. A vector $y$ satisfies these properties if and only if:
\begin{itemize}
\item there are exactly $s$ integers $i$ with $y[i]=1$ and $x[i]=1$, and
\item there are exactly $k-s$ integers $i$ with $y[i]=1$ and $x[i]=0$. 
\end{itemize}
So there are $\binom{|x|}{s} \binom{n-|x|}{k-s}$ such choices of $y$. The total probability is hence
\begin{align*}
\frac{1}{2^n} ~\sum_{k= (1/2 - a)n}^{(1/2 + a)n}~ \sum_{s=0}^{bn} \binom{|x|}{s} \binom{n-|x|}{k-s}
& = \frac{1}{2^n}~ \sum_{k= (1/2 - a)n}^{(1/2 + a)n}~ \sum_{s=0}^{bn}~   \binom{|x|}{s} \binom{n-|x|}{k-s} \\
& \leq \frac{1}{2^n}~ \sum_{k= (1/2 - a)n}^{(1/2 + a)n} ~\sum_{s=0}^{bn}~  \binom{(1/2+a)n}{s} \binom{(1/2+a)n}{k-s} \\
&\leq \frac{1}{2^n}~ \sum_{k= (1/2 - a)n}^{(1/2 + a)n}~(bn+1)\cdot \binom{(1/2+a)n}{bn} \binom{(1/2+a)n}{(1/2-a-b)n}. ~\text{(*)}
\end{align*}

Recall that if $k_1 < k_2 < n/3$ then $\binom{n}{k_1} < \binom{n}{k_2}$, and if $k_3 > k_4 > n/2$, then $\binom{n}{k_3} < \binom{n}{k_4}$. Step (*) therefore follows since $s \leq bn < \frac12 (1/2+a)n$ and $k-s \geq(1/2 - a - b)n > \frac12 (1/2 + a)n$ whenever $0<a,b<1/5$. Let $g(n) = (2an+1) \cdot (bn+1)$. Simplifying further, the above expression is at most 
\begin{align*}
&\frac{g(n)}{2^n} \binom{(1/2+a)n}{bn} \binom{(1/2+a)n}{(2a+b)n} \\
&\leq \frac{g(n)}{2^n} \cdot 2^{(1/2+a)n \cdot H(b/(1/2+a))} 2^{(1/2+a)n \cdot H((2a+b)/(1/2+a))} ~~~ \text{(by \eqref{entropy})} \\
&\leq \frac{g(n)}{2^n} \cdot 2^{(1/2+a)n \cdot H(2b)} 2^{(1/2+a)n \cdot H(4a+2b)} \\
&\leq \frac{g(n)}{4^n} \cdot 2^{(1/2+a)n \cdot 2 \cdot 2b \cdot \log(1/2b)} 2^{(1/2+a)n \cdot 2 \cdot (4a+2b) \cdot \log(1/(4a+2b))} ~~\text{($H(\eps) \leq 2\eps\log_2(1/\eps)$ for $\eps < 1/2$)}\\
&\leq \frac{g(n)}{2^n} \cdot 2^{f(a,b)n},
\end{align*}
where $f(a,b) = (1/2+a)(4b \log(1/2b) + (8a+4b)\log(1/(4a+2b)))$.
\end{proof}

Our third lemma is a simple linear-algebraic observation: given a low-rank matrix $M$ that computes another matrix $N$ on all but a small number of rows and columns, $N$ must also have relatively low rank. 

\begin{lemma}\label{row-col-correction} Let $M'$ be a matrix of rank $r$ which is equal to $M$ except in at most $k$ columns and $\ell$ rows. Then the rank of $M$ is at most $r+k+\ell$.
\end{lemma}

\begin{proof}
We will start with $M'$, and add at most $k + \ell$ rank-one matrices to $M'$ so that it equals $M$.

Consider a column $c$ on which $M$ does not equal $M'$. We can add to $M'$ a correction matrix $C_c$ given by
$$
C_c(i,j)=
\begin{cases}
M(i,v) - M'(i,v) & \text{if } j=v, \\
0 & \text{otherwise.}
\end{cases}
$$
Then, $M' + C_c$ equals $M$ on column $c$, and is unchanged in any other column. Moreover, since $C_c$ is only nonzero on a single column, it has rank one. So all we have to do is add the correction matrix $C_c$ for each column $c$ on which $M$ and $M'$ differ. The rows of $M'$ can be corrected analogously.
\end{proof}

\begin{corollary}\label{correct-WH} Let $T$ be any $2^n \times 2^n$ matrix. Let $a \in (0,1/2)$, and let $M$ be a $2^n \times 2^n$  matrix of rank $r$, indexed by $n$-bit vectors. There is a $2^n \times 2^n$ matrix $M'$ of rank at most $r + 4 \cdot n\cdot 2^{n-\Omega(a^2 n)}$ such that $M'(v_i,v_j) = T(v_i,v_j)$ on all $v_i, v_j \in \{0,1\}^n$ where at least one of the following holds:
\begin{itemize}
\item $|v_i| \notin [(1/2-a)n, (1/2+a)n]$,
\item $|v_j| \notin [(1/2-a)n, (1/2+a)n]$, or,
\item $M(v_i,v_j) = T(v_i,v_j)$.
\end{itemize}
\end{corollary}

\begin{proof} The number of $v_i \in \{0,1\}^n$ with $|v_i| \notin [(1/2-a)n, (1/2+a)n]$ is at most \[\sum_{i=0}^{(1/2-a)n} \binom{n}{i} +\sum_{i=(1/2+a)n}^{n} \binom{n}{i} = 2 \sum_{i=0}^{(1/2-a)n} \binom{n}{i} \leq n \cdot 2^{n-\Omega(a^2 n)},\] by \eqref{entropy-half}. Applying Lemma~\ref{row-col-correction} to $M$ and $M'$ with $k$ and $\ell$ set to $2 \cdot n \cdot 2^{n-\Omega(a^2 n)}$, the result follows.  
\end{proof}

Let us outline how we'll use all of the above. First, we construct a matrix $M$ of rank about $2^{n-\Omega(\eps^2 n)}$ approximating $H_n$, using the polynomial from Lemma~\ref{IP2-polynomial} in a straightforward way. This matrix $M$ has far more erroneous entries than what we desire. But by Lemma~\ref{small-IP}, we can infer that the errors in $M$ are highly concentrated on a relatively small number of rows and columns. Applying Corollary~\ref{correct-WH}, the rows and columns can be ``corrected'' in a way that increases the rank of $M$ by only $2^{n-\Omega(\eps^2 n)}$. By Lemma~\ref{small-IP}, each row of the matrix left over will have $2^{O(\eps\log(1/\eps)n)}$ erroneous entries.

\begin{proof}[Proof of Theorem \ref{nonrigid-IP2}] In fact, we prove that one only has to modify $2^{O(\eps\log(1/\eps)n)}$ entries in each row of $H_n$, to obtain the desired rank.

Let $\eps > 0$ be given. By Lemma~\ref{IP2-polynomial}, there is a polynomial $p(x,y)$ in $2n$ variables with $m = 2^{n-\Omega(\eps^2 n)}$ monomials which computes $(-1)^{\langle x,y\rangle}$ correctly, on all $(x,y) \in \{0,1\}^{2n}$ such that $\langle x,y\rangle \in [2 \eps n, (1/2 + \eps) n]$.

Construct a $2^n \times 2^n$ matrix $M$ of rank $m$ as in Corollary \ref{lemma:monomial-to-rank}, so that $M(x,y) = p(x,y)$. By definition, $M$ equals $H_n$ on all $(x,y) \in \{0,1\}^{2n}$ satisfying $\langle x,y\rangle \in [2 \eps n, (1/2 + \eps) n]$.

Applying Corollary~\ref{correct-WH} to $M$ with $T=H_n$ and $a=\eps$, we obtain a matrix $M'$ of rank $m+4 \cdot n \cdot 2^{n-\Omega(\eps^2 n)}$ which is correct on all $(x,y)$ where either $|x| \notin [(1/2-\eps)n, (1/2+\eps)n]$, $|t| \notin [(1/2-\eps)n, (1/2+\eps)n]$, or $\langle x,y\rangle \in [2 \eps n, (1/2 + \eps) n]$.

Fix a row of $H_n$ indexed by $x \in \{0,1\}^n$ with $|x| \in [(1/2-\eps)n, (1/2+\eps)n]$ (note the other rows are already correct). To show that $M'$ differs from $H_n$ on a small number of entries, we need to bound the number of $y$ such that none of the above conditions hold, i.e., 
\begin{enumerate}
\item $|y| \in [(1/2-\eps)n, (1/2+\eps)n]$ and
\item $\langle x,y\rangle \notin [2 \eps n, (1/2 + \eps) n]$.
\end{enumerate}
Note for our given $x$, it is never true that $\langle x,y\rangle > (1/2 + \eps) n$. Therefore we only need to bound the number $N$ of $y$ such that $|y| \in [(1/2-a)n, (1/2+a)n]$ and yet $\langle x,y\rangle < 2 \eps n$. By Lemma~\ref{small-IP} with $a=\eps$ and $b=\eps$, the probability that a random $y$ satisfies $\langle x,y\rangle < 2 \eps n$ and $|y| \in [(1/2-\eps)n, (1/2+\eps)n]$, is at most $O(n^2) \cdot 2^{(f(\eps,\eps) - 1)n}$, where $f \to 0$ as $\eps \to 0$. Therefore $N \leq 2^{n} \cdot O(n^2) \cdot 2^{(f(\eps,\eps) - 1)n} \leq O(n^2) \cdot 2^{f(\eps,\eps)n}$. 

Now for sufficiently large $n$ and $\eps \in (0,1/2)$, $M'$ has rank at most $m + 4 \cdot n \cdot 2^{n-\Omega(\eps^2 n)}\leq 5n \cdot 2^{n-\Omega(\eps^2 n)}$. Furthermore, on every row, $M'$ differs from $H_n$ in at most $n^2 \cdot 2^{f(\eps,\eps)}\leq 2^{O(\eps\log(1/\eps)n)}$ entries. 
\end{proof}

Other rigidity upper bounds for $H_n$ are described in Appendix~\ref{sec:high-error-WH}.

\section{Probabilistic Rank and Rigidity: An Equivalence}
\label{sec:equivalence}

In this section, we show that the probabilistic rank of $H_n$ and the rigidity of $H_n$ are the \emph{same} concept over fields. It is easy to see that if $\eps$-probabilistic rank of $H_n$ is $k$ over a field $K$, then the rank-$k$ rigidity of $H_n$ is at most $\eps 2^{2n}$ over $K$. Exploiting the random self-reducibility of the $H_n$ function, we can show a converse: lower bounds on probabilistic rank imply proportionate rigidity lower bounds. This is of interest because probabilistic rank lower bounds appear to be fundamentally easier to prove than rigidity lower bounds. 

\begin{reminder}{Theorem~\ref{prob-rank-rigidity-IP2}} For every field $K$ and for every $n$, ${\cal R}_{H_n}(r) \leq \eps 2^{2n}$ over $K$ if and only if $H_n$ has $\eps$-probabilistic rank $r$ over $K$.
\end{reminder}

First let us give some definitions. Let $\otimes$ denote the outer product of vectors. For vectors $a \in K^{2^n}$ whose entries are indexed by $v_1, \ldots, v_{2^n} \in \{0,1\}^n$, and $x,y \in \{0,1\}^n$, let $a^{(x,y)}$ denote the vector in $K^{2^n}$ given by
$$a^{(x,y)}[v_i] = (-1)^{\langle v_i, y \rangle} a[v_i \oplus x].$$
This permutes the entries of $a$, then negates half of the entries.

\begin{proof} One direction is trivial: low probabilistic rank implies low rigidity, by simply drawing a ``typical'' matrix from the distribution. For the other direction, suppose $a_1, \ldots, a_r$ and $b_1, \ldots, b_r$ are vectors in $K^{2^n}$ such that the $2^n \times 2^n$ matrix
\begin{align}\label{Mdef} M := \sum_{k=1}^r a_k \otimes b_k\end{align}
differs from $H_n$ in at most $\eps 2^{2n}$ entries. Pick vectors $x,y \in \{0,1\}^n$ uniformly at random, and consider the $2^n \times 2^n$ matrix
\begin{align}\label{Mpdef} M' = (-1)^{\langle x,y \rangle} \sum_{k=1}^r a^{(x,y)}_k \otimes b^{(y,x)}_k.\end{align}
In this form it is clear that $M'$ has rank at most $r$. We claim that each entry of $M'$ is equal to the corresponding entry of $H_n$ with probability at least $1 - \eps$, over the choice of $x$ and $y$, which will complete the proof.

Consider a given entry $M'(v_i, v_j)$. It is sufficient to show that if $M(v_i \oplus x, v_j \oplus y) = H_n(v_i \oplus x, v_j \oplus y)$ then $M'(v_i, v_j) = H_n(v_i, v_j)$, since $(v_i \oplus x, v_j \oplus y)$ is a uniformly random pair of vectors in $\{0,1\}^n$. Suppose this is the case, meaning $M(v_i \oplus x, v_j \oplus y) = (-1)^{\langle v_i \oplus x, v_j \oplus y \rangle}$.
Applying definition (\ref{Mdef}) and then (\ref{Mpdef}) we see that
\begin{align*}(-1)^{\langle v_i \oplus x, v_j \oplus y \rangle} &= \sum_{k=1}^r a_k[v_i \oplus x] \cdot b_k[v_j \oplus y] \\
&= (-1)^{\langle v_i, y \rangle + \langle v_j, x \rangle} \sum_{k=1}^r (-1)^{\langle v_i, y \rangle} a_k[v_i \oplus x] \cdot (-1)^{\langle v_j, x \rangle} b_k[v_j \oplus y] \\
&= (-1)^{\langle v_i, y \rangle + \langle v_j, x \rangle} \sum_{k=1}^r  a_k^{(x,y)} [v_i] \cdot b_k^{(y,x)} [v_j] \\
&= (-1)^{\langle v_i, y \rangle + \langle v_j, x \rangle} \cdot (-1)^{\langle x,y \rangle} \cdot M'(v_i, v_j).
\end{align*}
Rearranging, we see as desired that $$M'(v_i, v_j) = (-1)^{\langle v_i \oplus x, v_j \oplus y \rangle + \langle v_i, y \rangle + \langle v_j, x \rangle + \langle x,y \rangle} = (-1)^{\langle v_i , v_j \rangle},$$ where the last step follows from the bilinearity of the inner product $\langle \cdot, \cdot \rangle$.\end{proof}

Therefore, proving communication lower bounds for the IP2 function against (for example) the class $\BP \cdot \MOD_m \P$ is \emph{equivalent} to proving rigidity lower bounds for $H_n$ over the ring $\Z_m$. Applying our rigidity upper bounds for $H_n$ (Theorems~\ref{nonrigid-IP2} and \ref{high-error-IP2}), we obtain surprisingly low probabilistic rank bounds for $H_n$ (and therefore communication-efficient protocols as well):

\begin{corollary} For every field $K$, for every sufficiently small $\eps > 0$, and for all $n$, $H_n$ has $1/2^{n(1-\eps)}$-probabilistic rank at most $2^{n-\Omega(\eps^2/\log(1/\eps)) n}$ over $K$, and $\eps$-probabilistic rank at most $(1/\eps)^{O(\sqrt{n} \log n)}$.
\end{corollary}

Our reduction from rigidity to probabilistic rank in fact works for any (non-adaptive) random self-reducible function~\cite{fef93} that makes a small number of oracle calls. See Appendix~\ref{sec:general-self-reduction}.

\section{Discussion}

Our most significant finding is that Hadamard matrices are not as rigid as previously believed: for every $\eps > 0$, there are infinitely many $N$ and $N \times N$ Hadamard matrices whose rank drops below $N^{1-\Omega(\eps^2/\log(1/\eps))}$ after modifying only $N^{\eps}$ entries in each row. This rules out a proof of arithmetic circuit lower bounds for the DFT over $\Z_2^n$ via matrix rigidity. Our proof shows precisely how low rank-rigidity can be more powerful than low-sparsity polynomial approximations: we start with a sparse polynomial that has errors concentrated on negligibly many rows and columns, and use a simple lemma to correct most erroneous rows and columns.

Are there other conjectured-to-be-rigid matrices which are not? One candidate would be the generating matrix of a good linear code over $\F_2$. Very recently, Goldreich~\cite{Goldreich-Dvir16} has reported a distribution of matrices in which most of them are the generating matrix of a good linear code that is not rigid, found by Dvir. It would be very interesting to find an explicit code with this property. Another next natural target would be Vandermonde matrices. Given a field $\F$ of order $n$, and letting $g$ be a generator of the multiplicative group $\F^{\times}$, the $n \times n$ matrix $V[i,j] := g^{(i-1) \cdot (j-1)}$ also has structure that may be similarly exploitable.  

Our proof that functions with small $\THR \circ \THR$ circuits have low probabilistic sign-rank (Theorem~\ref{thm:thrthr-sign-rank}) effectively shows how to randomly reduce an ``inner product defined by a $\THR \circ \THR$'' to an ``inner product defined by a $\THR \circ \XOR$.'' It seems likely that this result could have further applications (beyond what we showed). The theorem suggests the research question: is it possible to write a $\THR \circ \THR$ circuit as an small ``approximate-MAJORITY'' of $\THR \circ \XOR$ circuits, i.e. a \emph{probabilistic} PTF, in the sense of~\cite{Alman-Chan-Williams16}? This would be an intriguing simulation of depth-two threshold circuits.

Another significant theme in this paper is the close relationship between probabilistic rank and threshold circuits, as well as rigidity. It seems likely that more algorithmic applications will be found by further study of probabilistic rank of matrices; perhaps some lower bounds can also be proved via these connections.

\paragraph{Acknowledgements.} We thank Mrinal Kumar for helpful comments, Oded Goldreich, Michael Forbes and Satya Lokam for helpful discussions, and Amir Shpilka and Avishay Tal for patiently listening to R.W.'s conjectures and results on non-rigidity at Banff (BIRS) in August 2016.

\bibliographystyle{alpha}
\bibliography{rigidity}

\appendix

\section{Rigidity Upper Bounds For High Error}\label{sec:high-error-WH}

In this section, we prove upper bounds on the rigidity of the Walsh-Hadamard transform in the regime where the error is constant, or much larger than $1/2^n$; this setting is of interest for communication complexity lower bounds.

\begin{reminder}{Theorem~\ref{high-error-IP2}} For every integer $r \in [2^{2n}]$, one can modify at most $2^{2n}/r$ entries of $H_n$ and obtain a matrix of rank at most $(n/\log(r))^{O(\sqrt{n \log(r)})}$.
\end{reminder}

The proof follows from applying an optimal-degree probabilistic polynomial for symmetric functions:

\begin{theorem}[\cite{JoshRyan}] \label{thm:sym-poly}
There is a probabilistic polynomial over any field, or the integers, for any symmetric Boolean function on $n$ variables, with error $\eps$ and degree $O(\sqrt{n \log(1/\eps)})$.
\end{theorem}

\begin{proofof}{Theorem~\ref{high-error-IP2}}
Set $\eps = 1/r$, and define the Boolean function $IP2 : \{0,1\}^{2n} \to \{-1,1\}$ by $IP2(x,y) = (-1)^{\langle x, y \rangle}$ for all $x,y \in \{0,1\}^n$. We can see that $H_n$ is the truth table matrix $M_{IP2}$. By Corollary \ref{cor:poly-to-rigid}, it is sufficient to construct a probabilistic polynomial for $IP2$ with error $\eps$ and $(n/\ln(1/\eps))^{O(\sqrt{n \log(1/\eps)})}$ monomials. Consider the Boolean function $PARITY(z_1,\ldots,z_n) = (-1)^{z_1 + \cdots + z_n}$ for all $z \in \{0,1\}^n$, and note that $IP2(x_1, \ldots,x_n,y_1,\ldots,y_n) = PARITY(x_1y_1, x_2y_2, \ldots, x_ny_n)$.
Since $PARITY$ is symmetric, by Theorem \ref{thm:sym-poly} it has a probabilistic polynomial $P$ of error $\eps$ and degree $d = O(\sqrt{n \log(1/\eps)})$. Hence, the distribution of $p(x_1y_1, \ldots, x_ny_n)$ over $p$ drawn from $P$ is a probabilistic polynomial for $IP2$. Since we are only interested in the value of $p(z)$ when $z \in \{0,1\}^n$, we can make $p$ multilinear by applying the equation $v^2 = v$ to all variables. Then the number of monomials of $p$ is at most $\sum_{i=0}^{O(\sqrt{n \log(1/\eps)})} \binom{n}{i} \leq (n/\ln(1/\eps))^{O(\sqrt{n \log(1/\eps)})}$. Since in $p(x_1y_1, \ldots, x_ny_n)$ we are substituting in a monomial for each variable, its expansion has the same number of monomials as $p$, as desired.
\end{proofof}

\section{Rigidity Upper Bound for SYM-AND circuits}\label{sec:sym-and}

Here we generalize Theorems \ref{nonrigid-IP2} and \ref{high-error-IP2} to $\SYM \circ \AND$ circuits. In the proof of Theorem~\ref{nonrigid-IP2}, the key property of the $IP2$ function required is that has the form 
$$IP2(x_1, \ldots, x_n, y_1, \ldots, y_n) = f(x_1 \wedge y_1, \ldots, x_n \wedge y_n),$$
where $f$ is a symmetric Boolean function (in our case, $f$ computes parity). The same proof yields the following generalization:

\begin{theorem} For every symmetric function $f : \{0,1\}^n \to R$, define the function $IP_f : \{0,1\}^{2n} \to R$ by $IP_f(x,y)= f(x_1 \wedge y_1, \ldots, x_n \wedge y_n)$. For all sufficiently small $\eps$, there is a $\delta < 1$ and a matrix of rank $2^{\delta n}$ which differs from the truth table matrix $M_{IP_f}$ in at most $2^{(1+\eps)n}$ entries.
\end{theorem}

The proof of Theorem~\ref{high-error-IP2} only requires a probabilistic polynomial construction in Corollary \ref{cor:poly-to-rigid}. Our probabilistic matrix distribution simply substitutes monomials into the probabilistic polynomial of Theorem~\ref{thm:sym-poly} for any symmetric function. Since each monomial can be viewed as an $\AND$, the same argument will work for any $\SYM \circ \AND$ circuit.

\begin{theorem} For any Boolean function $f : \{0,1\}^{2n} \to R$ which can be written as a $\SYM \circ \AND$ circuit with $s$ $\AND$ gates, and for every integer $r \in [2^{2n}]$, one can modify $2^{2n}/r$ entries of the truth table matrix $M_f$ and obtain a matrix of rank at most $(s / \log r)^{O(\sqrt{s \log r})}$.
\end{theorem}

\section{Explicit Rigid Matrices and Threshold Circuits}\label{sec:rigid-circuit-lbs}
In this section, we show how explicit rigidity lower bounds would also imply circuit lower bounds where we currently only know weak results (e.g., we know that some functions in $\E^{\NP}$ do not have such circuits). 

\begin{theorem}\label{thm:ATAT-rigidity}
For every constant $\delta>0$ and every $\AC^0 \circ \THR \circ \AC^0 \circ \THR$ circuit $C$ of size-$s = n^{2-\delta}$ and depth-$d = o(\log(n)/\log\log(n/\eps))$, there exists a $\gamma > 0$ such that the truth table of $C$ as a $2^{n/2}\times 2^{n/2}$ matrix $M_C$ has rigidity ${\cal R}_{M_C}\left(2^{n^{1-\gamma} \log(1/\eps)}\right) \leq \eps 2^n$, for all $\eps \in (1/2^n,1)$, over any field.
\end{theorem}

Our proof will use a technique by Maciel and Therien for converting each middle layer $\THR$ gate into an equivalent $\AC^0 \circ \MAJ$ circuit:

\begin{theorem}[\cite{Maciel-Therien98} Theorem 3.3, \cite{Alman-Chan-Williams16} Theorem 7.1] \label{thm:thr-to-maj}
For every $\alpha>0$, every $\THR$ on $n$ inputs can be computed by a polynomial-size $\AC^0 \circ \MAJ$ circuit where the fan-in of each $\MAJ$ gate is $n^{1+\alpha}$ and the circuit has depth $O(\log(1/\alpha))$.
\end{theorem}

We will also use Tarui's probabilistic polynomial for $\AC^0$:

\begin{theorem}[\cite{Tarui93} Theorem 3.6] \label{thm:tarui-ac0}
Every circuit in $\AC^0$ with depth $d$ has a probabilistic polynomial over $\mathbb{Z}$ of degree $O(\log^{d}(n))$ and error $1/2^{\log^{O(1)}(n)}$.
\end{theorem}

\begin{proofof}{Theorem~\ref{thm:ATAT-rigidity}}
By Lemma \ref{lem:ltf-rank}, each $\THR$ gate in the bottom layer has $\eps/s$-probabilistic rank $O(n^2 s/\eps)$. We will design a probabilistic polynomial for the upper $\AC^0 \circ \THR \circ \AC^0$ circuitry, which will give the desired result when composed with this probabilistic rank expression.

First, each $\THR$ gate in the middle layer has fan-in at most $s = n^{2-\delta}$. Applying Theorem \ref{thm:thr-to-maj} with $\alpha = \delta/2$ to each, the upper $\AC^0 \circ \THR \circ \AC^0$ circuit becomes a $\AC^0 \circ \MAJ \circ \AC^0$ where each $\MAJ$ gate has fan-in at most $n^{(2-\delta)(1+\delta/2)}=n^{2-\delta^2/2}$, and the depth is still $O(d)$.

We can now apply the probabilistic polynomial for $\AC^0$ from Theorem \ref{thm:tarui-ac0} with degree $O(\log^{d}(n))$ error $1/2^{\log^{O(1)}(n)}$ to the $\AC^0$ circuits, and the probabilistic probabilistic polynomial for symmetric functions on $n^{2-\delta^2/2}$ bits from Theorem \ref{thm:sym-poly} with error $\eps/s$ and degree $O(n^{1 - \delta^2/4} \log(s/\eps))$ to the $\MAJ$ gates in the middle. This results in a probabilistic polynomial of degree $O(n^{1 - \delta^2/4} \log^{O(d)}(n/\eps))$. For $d = o(\log(n)/\log\log(n/\eps))$, this is $O(n^{1 - \beta})$ for any $\beta \in (0,\delta^2/4)$.

We can view the terms in the probabilistic rank expression for the $\THR$ gates in the bottom layer as variables that we substitute into this probabilistic polynomial; the number of monomials in this expansion will upper bound the rank, as in Lemma \ref{lemma:monomial-to-rank}. Since there are at most $s$ such gates, and each probabilistic rank expression has $O(n^2 s/\eps)$ terms, we are substituting $O(n^2s^2/\eps)$ terms into our polynomial. Hence, the number of monomials will be upper bounded by
$$(n^2 s^2/\eps)^{O(n^{1 - \beta})} = 2^{O(n^{1 - \gamma}) \log(1/\eps)},$$
for any $\gamma < \beta$. This is of the desired form, where we can pick any positive value $\gamma < \delta^2/4$. The correctness follows by union bounding over all $\leq s$ probabilistic substitutions we make, each of which has error probability at most $\eps/s$.
\end{proofof}

From the above theorem, setting the error $\eps$ appropriately, we infer a new consequence of explicit rigid matrices:

\begin{reminder}{Theorem~\ref{THR-circuits-rigidity}} Let $K$ be an arbitrary field, and $\{M_n\}$ be a family of Boolean matrices such that (a) $M_n$ is $n\times n$, (b) there is a $\poly(\log n)$ time algorithm $A$ such that $A(n,i,j)$ prints $M_n(i,j)$, and (c) there is a $\delta > 0$ such that for infinitely many $n$, \[{\cal R}_{M_n}\left(2^{(\log n)^{1-\delta}}\right) \geq \frac{n^2}{2^{(\log n)^{\delta/2}}} \text{ over $K$}.\] 
Then the language $\{(n,i,j)\mid M_n(i,j)=1\} \in \P$ does not have $\AC^0 \circ \THR \circ \AC^0 \circ \THR$ circuits of $n^{2-\eps}$-size and $o(\log n/\log \log n)$-depth, for all $\eps > 0$.
\end{reminder}

Therefore, proving strong rigidity lower bounds for explicit matrices over $\R$ has consequences for Boolean circuit complexity as well. Indeed, the desired circuit lower bounds could be derived from lower-bounding probabilistic rank.

\section{Sign-rank Rigidity and Depth-Two Threshold Circuits}\label{sec:sign-rank}

Given a matrix $A \in \R^{n\times n}$, its sign rank is the minimum rank of any $B \in \{-1,1\}^{n\times n}$ such that $\sgn(A[i,j])=\sgn(B[i,j])$ for all $i,j \in [n]$. The $\eps$-probabilistic sign-rank of $A$ is defined analogously as with probabilistic rank. We say $A$ has \emph{sign rank $r$-rigidity $t$} if a minimum of $t$ entries of $A$ need to be modified in order for $A$ to have sign rank at most $r$. 

First, we observe (in Appendix~\ref{sec:random-sign-rigid}) that in the sign-rank setting, random -1/1 matrices are still rigid: for example, with high probability, a random -1/1 matrix has sign-rank-$(n/\log^2 n)$ rigidity at least $\Omega(n^2)$. Even though most -1/1 matrices have high sign-rank rigidity, we show that the truth table of a small $\THR \circ \THR$ circuit is always close to a matrix of low sign-rank. For even $n$, we say a function $f: \{0,1\}^n \rightarrow \{0,1\}$ has \emph{$\eps$-probabilistic sign rank $r$} if the truth table of $C$ construed as a $2^{n/2} \times 2^{n/2}$ matrix has $\eps$-probabilistic sign-rank $r$.

\begin{theorem} \label{thm:thrthr-sign-rank} For every function $f$ with a $\THR \circ \THR$ circuit of size $s$, and every $\eps > 0$, the $\eps$-probabilistic sign-rank of $f$ is $O(s^2 n^2/\eps)$. Moreover, we can sample a low-rank matrix from the distribution of matrices in $2^{n/2} \cdot \poly(s,n)$ time.
\end{theorem}

We will prove this theorem in a few steps. Let $EQ_n : \{0,1\}^{2n} \to \{0,1\}$ be the equality function, i.e., $EQ_n(x,x) = \left[x=y\right]$ (using Iverson bracket notation). Similarly, let $LEQ_n : \{0,1\}^{2n} \to \{0,1\}$ be the function $LEQ_n(x,x) = \left[x\leq y\right]$ where $x$ and $y$ are interpreted as integers in $\{0,\ldots,2^n-1\}$. 

\begin{lemma} \label{lem:eq-rank} For every $n$, $EQ_n$ has $\eps$-probabilistic rank at most $O(1/\eps)$ over any field.
\end{lemma}

\begin{proof} We mimic a well-known randomized communication protocol for $EQ_n$. Pick $k=\lceil \log_2(1/\eps) \rceil$ uniformly random subsets $S_1, \ldots, S_k \subseteq \{0,1\}^n$, and define the hash functions $h_1, \ldots, h_k : \{0,1\}^n \to \{0,1\}$ by $h_i(x) = \bigoplus_{j \in S_i} x_j$. 
Note that $h_i(x) \neq h_i(y)$ with $1/2$ chance if $x \neq y$. Hence, the following expression equals $EQ(x,y)$ with error probability at most $\eps$:
\begin{align}\label{eq-expression} \prod_{i=1}^k (h_i(x) h_i(y) + (1-h_i(x)) (1-h_i(y))).\end{align}
When expanded out, (\ref{eq-expression}) is a sum of $2^k = O(1/\eps)$ terms of the form $f(x) \cdot g(y)$ for some functions $f$ and $g$, each of which has rank one.
\end{proof}

\begin{lemma} \label{lem:leq-rank} For every $n$, $LEQ_n$ has $\eps$-probabilistic rank at most $O(n^2/\eps)$ over any field.
\end{lemma}

\begin{proof} 
We express $LEQ_n$ in terms of $EQ$ predicates which check for the first bit in which $x$ and $y$ differ, as
\begin{align}\label{leq-expression} LEQ_n(x_1, \ldots, x_n, y_1, \ldots, y_n) = \sum_{i=1}^n (1 - x_i) \cdot y_i \cdot EQ_{i-1}(x_1, \ldots,x_{i-1},y_1,\ldots,y_{i-1}).\end{align}
We then get the desired rank bound by replacing each $EQ$ with the probabilistic rank expression from Lemma \ref{lem:eq-rank} with error $\eps/n$. By the union bound, all $n$ of the $EQ$ predicates will be correct with probability at least $1 - \eps$, and hence we will correctly compute $LEQ_n$.
\end{proof}

\begin{lemma} \label{lem:ltf-rank} For every $n$, every linear threshold function $f : \{0,1\}^{2n} \to \{0,1\}$ has $\eps$-probabilistic rank $O(n^2/\eps)$.
\end{lemma}

\begin{proof} A linear threshold function $f$ is defined as $f(x_1,\ldots,x_n,y_1,\ldots,y_n) = \left[\sum_i v_i x_i + \sum_i w_i y_i \geq k\right]$, where all $v_i$'s, $w_i$'s, and $k$ are reals. We want to show that the $2^n \times 2^n$ matrix indexed by $x_i$-assignments on the rows and $y_i$-assignments on the columns has low probabilistic rank.
We will exploit the fact that the linear forms on $x_i$'s and $y_i$'s can be preprocessed separately in a rank decomposition.

Define $a : \{0,1\}^n \to \mathbb{R}$ by $a(x_1, \ldots, x_n) = \sum_{i=1}^n v_i x_i$, and $b : \{0,1\}^n \to \mathbb{R}$ by $b(y_1, \ldots, y_n) = k - \sum_{j=1}^n w_i y_i$. Hence \[f(x,y)= \left[a(x) \leq b(y)\right].\]
Let $L$ be the list, sorted in increasing order, of all values of $a(x)$ and $b(y)$, for all $x \in \{0,1\}^n$ and $y \in \{0,1\}^n$. Then define the function $\alpha : \{0,1\}^n \to \{0,1\}^{n+1}$ where $\alpha(x)$ equals the earliest index of $a(x)$ in the sorted list $L$, interpreted as a $n+1$ bit number. Define $\beta : \{0,1\}^n \to \{0,1\}^{n+1}$ similarly. Then \[f(x,y) = LEQ_{n+1}(\alpha(x), \beta(y)).\] So the $\eps$-probabilistic rank of $M_f$ is at most that of $M_{EQ_{n+1}}$, which we upper-bounded in Lemma \ref{lem:leq-rank}.
\end{proof}

Now we are ready to upper-bound the probabilistic sign-rank of depth-two threshold circuits:

\begin{proofof}{Theorem~\ref{thm:thrthr-sign-rank}}
We interpret our $\THR \circ \THR$ circuit $C$ as a function on two groups of $n/2$ bits, $x_1, \ldots, x_{n/2}$, and $y_1, \ldots, y_{n/2}$. Let $w_1, \ldots, w_s, k \in \mathbb{R}$ be the weights of the output gate, so that \[C(x_1,\ldots,x_n,y_1,\ldots,y_n) = \left[\sum_{i=1}^s w_i \cdot f_i(x_1,\ldots,x_n,y_1,\ldots,y_n) \leq k\right],\] for $s$ different $\THR$ functions $f_i$. By Lemma \ref{lem:ltf-rank}, the truth table matrix $M_{f_i}$ of each $f_i$ has $(\eps/s)$-probabilistic rank $r = O(n^2 s / \eps)$. Our probabilistic distribution of matrices for $M_C$ can be constructed as follows: for all $i=1,\ldots,s$, draw a random rank-$r$ matrix $P_{i}$ from the distribution for $M_{f_i}$, and set
\[Q_C = (-k \cdot I) + \sum_i (w_i \cdot P_{i}).\] $Q_C$ has rank at most $s r + 1 \leq O(n^2 s^2 / \eps)$ and for all $(\vec{x},\vec{y})$,  $\Pr[\text{sign}(Q_C[\vec{x},\vec{y}]) \neq C(\vec{x},\vec{y})] \leq \eps$.
\end{proofof}

Are there explicit matrices with non-trivial sign-rank rigidity? We observe that the best-known rank rigidity lower bounds for $H_n$ extend to sign-rank rigidity:

\begin{theorem}[Follows from Razborov and Sherstov~\cite{Razborov-Sherstov10}]
For all $n$, and $r \in [2^{n/2},2^n]$, the sign-rank-$r$ rigidity of $H_n$ is at least $\Omega(4^n/r)$.
\end{theorem}

\begin{proof} Theorem 5.1 of~\cite{Razborov-Sherstov10} gives the following lower bound on sign-rank: given any matrix $A \in \{-1,1\}^{n \times n}$, suppose that all but $h$ entries of matrix $\tilde{A}$ have absolute value at least $\gamma$. Then \[\text{sign-rank}(A) \geq \frac{\gamma n^2}{||A||n + \gamma h},\] where $||A||$ is the spectral norm of $A$. For the case of $H_n$, if we modify $h := 4^n/r$ entries arbitrarily, all but $h$ entries have absolute value equal to $1$. Thus 
\[\text{sign-rank}(H_n) \geq \frac{4^n}{||H_n||n + 4^n/r}.\] As $||H_n||\leq O(2^{n/2})$~\cite{forster02}, we have $\text{sign-rank}(H_n)\geq \Omega(4^n/(2^{3n/2} + 4^n/r)) \geq \Omega(2^{n/2}+r) \geq \Omega(r)$. 
\end{proof} 

Can the above lower bound be improved slightly? Combining the previous two theorems, it follows that any minor improvement in the above rank/rigidity trade-off would begin to imply lower bounds for $\THR \circ \THR$:

\begin{theorem}\label{sign-rank-rigid-IP2} Suppose there is an $\alpha > 0$ such that for infinitely many $n$, the sign rank $r$-rigidity of $H_n$ is $\Omega(4^n/r^{1-\alpha})$, for some $r \geq \omega(n^{2/\alpha}s(n)^{2/\alpha})$. Then the Inner Product Modulo $2$ does not have $\THR \circ \THR$ circuits of $s(n)$ gates. 
\end{theorem}

\begin{proof} Suppose the sign rank $r$-rigidity of $H_m$ is $\Omega(4^n/r^{1-\alpha})$. Let $\eps = 1/r^{1-\alpha}$. 
It follows that the $\eps$-probabilistic sign-rank of $H_n$ is greater than $r$. But for a $\THR \circ \THR$ function with $s$ gates, its matrix always has $\Omega(\eps)$-probabilistic rank $O(s^2 n^2/\eps) = O(s^2 n^2 r^{1-\alpha})$, by Theorem~\ref{thm:thrthr-sign-rank}. Thus we have a contradiction when $O(s^2 n^2 r^{1-\alpha})$ is asymptotically less than $r$, i.e., \[r = \omega(n^{2/\alpha} s^{2/\alpha}),\] corresponding to an $s$-gate lower bound against $\THR \circ \THR$ circuits. Since $H_n$ is just a linear translation of the matrix for Inner Product Modulo $2$, the proof is complete.
\end{proof}

For instance, proving the sign-rank $2^{\alpha n}$-rigidity of $H_n$ is at least $4^n/2^{.999\alpha n}$ for some $\alpha > 0$ would imply exponential-gate lower bounds for depth-two threshold circuits computing IP2.

\section{Random Matrices are Sign-Rank Rigid}
\label{sec:random-sign-rigid}

The proof that random -1/1 matrices have high sign-rank  rigidity follows readily from recent work:   

\begin{theorem}[Follows from Alon-Moran-Yehudayoff~\cite{AlonMY16}]Let $r(n) = o(n/\log n)$. For all sufficiently large $n$, a random $n \times n$ matrix with $-1/1$ entries has sign-rank-$r(n)$ rigidity at least $\Omega(n^2)$, with high probability.
\end{theorem}

\begin{proof} There are $2^{n^2}$ matrices over $\{-1,1\}$. The number of distinct matrices with sign rank at most $r$ is bounded by $2^{O(r n \log n)}$~\cite{AlonMY16}. For a fixed matrix $M$, the number of matrices within Hamming distance $d$ of $M$ is at most $O(\binom{n^2}{t})$. Thus the number of matrices for which up to $t$ entries can be changed to obtain a matrix of sign rank at most $r$, is upper-bounded by \[2^{O(r n \log n)} \cdot \binom{n^2}{t} \leq n^{O(rn)}\cdot (en^2/t)^t.\] Suppose we set $t = \eps n^2$. Then the above quantity is at most \[n^{O(rn)} \cdot (e/\eps)^{\eps n^2}.\] For $r = o(n/\log n)$ and $\eps \log_2(e/\eps) < 1$, a random matrix is not among these matrices with high probability. Therefore a random matrix has sign rank-$o(n/\log n)$ rigidity $\Omega(n^2)$ with high probability.
\end{proof}

\section{Equivalence Between Probabilistic Rank Modulo m and BP-MODm Communication Complexity}

Here we sketch how probabilistic rank over $\Z_m$ is equivalent to $\BP \cdot \MOD_m \P$ communication complexity:

\begin{proposition} Let $m > 1$ be an integer, let $f : \{0,1\}^n \times \{0,1\}^n \rightarrow \{0,1\}$, and let $M_f$ be its truth table matrix. Let $C_{\eps}(f)$ be the $\BP \cdot \MOD_m \P$ communication complexity of $f$ with error $\eps$, and let $\eps\text{-rank}_{\Z_m}(M_f)$ be the $\eps$-probabilistic rank of $M_f$ over $\Z_m$. Then $C_{\eps}(f) \leq \log_2(\eps\text{-rank}_{\Z_m}(M_f)+1) \leq 2C_{\eps}(f)$.
\end{proposition}

This proposition is different from the one quoted in the introduction (giving constant-factor equivalences between the log of the rank and the communication complexity) because we are assuming a more stringent communication model here. However, the more general model is often taken as the definition, in which case the probabilistic rank and communication complexity truly coincide.

First, given a distribution of low-rank matrices for $M_f$, it is easy to construct a protocol for $f$: Alice and Bob publicly randomly sample a matrix from the distribution, which is a product of two matrices $A$ and $B$. Alice takes the row of $A$ of length $r$ corresponding to her input, Bob takes the column of $B$ of length $r$ corresponding to his, and they then compute the inner product of these two vectors over $\Z_m$ with $\lceil \log_2(r+1)\rceil$ communication in the $\MOD_m \P$ model.

To construct a distribution of matrices from communication protocols, we do a simple modification of the $\BPP^{\oplus \P}$ communication model. In fact, sometimes the literature \emph{defines} the $\BPP^{\oplus \P}$ communication model in this modified way~\cite{GoosPW16}. After the public randomness is chosen, Alice and Bob can, along with their $c$ nondeterministic bits, also sum over all possible \emph{transcripts} of at most $c$ bits between them. For each choice of randomness and nondeterminism there is a unique accepting transcript, so this extra choice does not alter the number of accepting communication patterns. But in this modified version, now Alice and Bob do not even have to communicate: they only have to send a single bit indicating whether they would accept or not, given the transcript and the nondeterminism. From such a protocol, it is straightforward to construct a $2^n \times 2^{2c}$ matrix $A$ representing Alice's protocol and a $2^{2c} \times 2^n$ matrix $B$ representing Bob, for any given string of public randomness. 

\section{Random Self-Reducibility, Rigidity, and Probabilistic Rank} 
\label{sec:general-self-reduction}

The reduction from rigidity to probabilistic rank works for any (non-adaptive) random self-reducible function~\cite{fef93} that makes a small number of oracle calls. Our notion of random self-reducibility is adapted for the communication complexity setting (for example, we do not care about the feasibility of the reduction).

\begin{definition} A function $f:\{0,1\}^n \times \{0,1\}^n \rightarrow \{0,1\}$ is \emph{$k$-random self-reducible} if there are random sampling procedures $S_1, S_2$ and  a function $g: \{0,1\}^k \rightarrow \{0,1\}$ such that:
\begin{itemize}
\item[(a)] $S_1$ takes $x \in \{0,1\}^n$ and a random string $r$ and outputs $x_1,\ldots,x_k \in \{0,1\}^n$ such that for all $n$-bit strings $z$, $\Pr_r[x_i = z]=1/2^n$ for all $i$,
\item[(b)] $S_2$ takes $y \in \{0,1\}^n$ and a random string $s$ and outputs $y_1,\ldots,y_k \in \{0,1\}^n$ such that for all $n$-bit strings $z$, $\Pr_s[y_i = z]=1/2^n$ for all $i$, and
\item[(c)] $f(x,y)=g(f(x_1,y_1),\ldots,f(x_k,y_k))$. 
\end{itemize}
\end{definition}

Requirements (a) and (b) in the definition ensures that each $x_i$ and $y_i$ are uniform random variables; requirement (c) says that we can reconstruct $f(x,y)$ from the values $f(x_1,y_1),\ldots,f(x_k,y_k)$. 

\begin{theorem} Let $r,n \in {\mathbb N}$ and $\eps \in (0,1)$, and let $f:\{0,1\}^{2n} \rightarrow \{0,1\}$ be $k$-random self-reducible. Suppose $M_f$ has rank-$r$ rigidity at most $\eps 4^n$ over $K$. Then the $(k\eps)$-probabilistic rank of $M_f$ over $K$ is at most $O(\binom{kr}{k})$.
\end{theorem}

\begin{proof} Suppose there is an $2^n \times r$ matrix $A$ and $r \times 2^n$ matrix $B$, such that $M_f$ and $A \cdot B$ differ in at most $\eps 4^n$ entries. We construct a distribution of low-rank matrices for $M_f$ as follows. 

Let $P(z_1,\ldots,z_k)$ be the unique multilinear polynomial over $K$ that represents the function $g$ from the random self-reduction for $f$. Given $k$ rows $X_1,\ldots,X_k \in K^r$ of $A$, and $k$ columns $Y_1,\ldots,Y_k \in K^r$ of $B$, define a polynomial in $2kr$ variables:
\[Q(X_1,\ldots,X_k,Y_1,\ldots,Y_k) = P(\ip{X_1}{Y_1},\ldots,\ip{X_k}{Y_k}).\] Treating each term of the form $X_i[j]\cdot Y_i[j]$ as a variable, $Q$ can be written as a sum of $t\leq\binom{kr}{k}$ total terms. Call the terms $m_1,\ldots,m_t$, each of which are over $2kr$ variables. 

Let $r$ be a random string for $S_1$ and $s$ be a random string for $S_2$. 
For $x \in \{0,1\}^n$, let $x_1,\ldots,x_k\in \{0,1\}^n$ be the outputs of $S_1(x)$ with randomness $r$. We define the $x$th row of a new $2^n \times t$ matrix $A_{r}$ to be 
\[ [m_1(A[x_1,:],\ldots,A[x_k,:],\vec{1},\ldots,\vec{1}),\ldots,m_t(A[x_1,:],\ldots,A[x_k,:],\vec{1},\ldots,\vec{1})].\] For $y \in \{0,1\}^n$, 
let $y_1,\ldots,y_k$ be the outputs of $S_2(x)$ with randomness $s$.
Define the $y$th column of a new $t \times 2^n$ matrix $B_{s}$ to be \[ [m_1(\vec{1},\ldots,\vec{1},B[:,y_1],\ldots,B[:,y_k]),\ldots,m_t(\vec{1},\ldots,\vec{1},B[:,y_1],\ldots,B[:,y_k])]^T.\] Then, for all $(x,y) \in \{0,1\}^n \times \{0,1\}^n$, the inner product of the $x$th row of $A_{r}$ and the $y$th column of $B_{s}$ is 
\begin{align*}
\sum_i m_i(A[x_1,:],\ldots,A[x_k,:],B[:,y_1],\ldots,B[:,y_k]) 
&= Q(A[x_1,:],\ldots,A[x_k,:],B[:,y_1],\ldots,B[:,y_k]) \\
&= P(\ip{A[x_1,:]}{B[:,y_1]},\ldots,\ip{A[x_k,:]}{B[:,y_k]}).
\end{align*}
Since $(A \cdot B)$ differs from $M_f$ on an $\eps$-fraction of entries, for uniform random $x_i, y_j \in \{0,1\}^n$ we have $\ip{A[x_i,:]}{B[:,y_i]} \neq f(x_i,y_i)$ with probability at most $\eps$. So with probability at least $1-k\eps$, $f(x_i,y_i) = \ip{A[x_i,:]}{B[:,y_i]}$ for all $i=1,\ldots,k$. Thus the polynomial $P(\ip{A[x_1,:]}{B[:,y_1]},\ldots,\ip{A[x_k,:]}{B[:,y_k]})$ being implemented by $A_r \cdot B_s$ outputs $f(x,y)$ with probability at least $1-k\eps$. Hence all matrices $C_{r,s} = A_{r} \cdot B_{s}$ in our defined distribution have rank at most $O(\binom{kr}{k})$, and for every $(x,y) \in \{0,1\}^n \times \{0,1\}^n$, $\Pr_{r,s}[C_{r,s}[x,y] = f(x,y)] \geq 1-k\eps$.\end{proof}

\end{document}